\documentclass[twoside]{article}
\usepackage[accepted]{aistats2019}

\usepackage{algorithm}
\usepackage{algpseudocode}
\usepackage[utf8]{inputenc}
\usepackage[T1]{fontenc}
\usepackage{hyperref}
\usepackage{amsmath, mathtools, amssymb}
\usepackage{amsthm}
\usepackage{booktabs}
\usepackage{graphicx}
\usepackage{caption}
\usepackage{subcaption}

\theoremstyle{plain}
\newtheorem{thm}{\protect\theoremname}
\theoremstyle{plain}
\newtheorem{prop}[thm]{\protect\propositionname}
\theoremstyle{remark}
\newtheorem*{rem}{\protect\remarkname}

\providecommand{\propositionname}{Proposition}
\providecommand{\remarkname}{Remark}
\providecommand{\theoremname}{Theorem}

\usepackage[style=numeric, backend=biber, maxcitenames=2, maxnames=2,maxbibnames=6, giveninits=true, uniquename=false,uniquelist=false,
			sortcites=true,doi=false,url=false, bibencoding=utf8]{biblatex} 

\DeclareLanguageMapping{english}{english-apa}
\renewbibmacro{in:}{}
\usepackage{csquotes}
\usepackage{microtype}

\begin{document}

\twocolumn[

\aistatstitle{Bernoulli Race Particle Filters}

\aistatsauthor{Sebastian M Schmon \And  Arnaud Doucet \And George Deligiannidis }

\aistatsaddress{Department of Statistics\\University of Oxford \And Department of Statistics\\University of Oxford \And Department of Statistics\\University of Oxford} ]
\begin{abstract}
When the weights in a particle filter are not available analytically,
standard resampling methods cannot be employed. To circumvent this
problem state-of-the-art algorithms replace the true weights with
non-negative unbiased estimates. This algorithm is still valid but at the cost
of higher variance of the resulting filtering estimates in comparison to a particle filter
using the true weights.
We propose here a novel algorithm that allows for resampling according to the true
intractable weights when only an unbiased estimator of the weights is available.
We demonstrate our algorithm on several examples.
\end{abstract}
{\small{}Keywords: }Bernoulli factory, random weight particle filter,
unbiased estimation, Rao-Blackwellization

\begin{refsection}
\section{Introduction}

Over the last 25 years particle filters have become a standard tool for optimal estimation in the context of general state space hidden Markov models (HMMs) with applications in ever more complex scenarios. HMMs are described by a latent (unobserved) process $(X_t)_{t\in\mathbb{N}}$
taking values in a space $\mathsf{X}$ and evolving according to the state transition density
\[
X_{t}\mid(X_{t-1}=x)\sim f(\,\cdot\mid x)
\]
for $t\geq2$ and $X_{1}\sim\mu(\cdot)$. The states can be observed through
an observation process $(Y_{t})_{t\in\mathbb{N}}$, taking values in a space $\mathsf{Y}$ with observation
density
\[
Y_{t}\mid(X_{t}=x)\sim g(\,\cdot\mid x).
\]
Particle filters (PFs) sequentially approximate the joint distributions
\begin{align*}
\pi_{t}\left(x_{1:t}\right) & = p(x_{1:t} \mid y_{1:t}) \\
 & = \frac{g(y_1\mid x_1)\mu(x_1)\prod_{k=2}^{t}g(y_{k}\mid x_{k})f(x_{k}\mid x_{k-1})}{p(y_{1:t})}\\
 & = \frac{g(y_1\mid x_1)\mu(x_1)\prod_{k=2}^{t}\varpi\left(x_{k}\mid y_{k},x_{k-1}\right)}{p(y_{1:t})}
\end{align*}
on the space $\mathsf{X}^t$ at time $t$ by evolving a set of samples, termed \emph{particles}, through time. We use here the shorthand notation $x_{j:k} = (x_j,\ldots,x_k)$. The quantity $p(y_{1:t})$ denotes the marginal likelihood of the model
\begin{equation}\label{eq:marginal_likelihood}
	p(y_{1:t}) = \int g(y_1 \mid x_1) \mu(x_1)\prod_{k=2}^{t}\varpi\left(x_{k}\mid y_{k},x_{k-1}\right) dx_{1:t}
\end{equation} 
which is usually intractable.
The PF algorithm proceeds as follows. Given a set of particles $\{X_{t-1}^i, i=1, \ldots, N\}$ at time $t-1$
the particles are propagated through $q(x_t \mid x_{t-1}, y_t)$. To adjust
for the discrepancy between the distribution of the proposed states and $\pi_{t}\left(x_{1:t}\right)$ the particles
are weighted by
\begin{align}\label{eq:bf_weights}
w_t(x_{t-1}, x_t) & =\frac{\varpi\left(x_{t}\mid y_{t},x_{t-1}\right)}{q(x_t\mid x_{t-1}, y_t)}, \quad t\geq 2, \\\notag
w_1(x_1)     & =\frac{g(y_1 \mid x_1)\mu(x_1)}{q(x_1\mid y_1)}.
\end{align}
A subsequent resampling step according to these weights ensures that only promising particles survive.
Here we consider only multinomial resampling, where we write $I_{1:N} \sim \textsf{Mult}\left\{N; w_{1}, \ldots, w_{N}\right\}$ to denote the vector $I_{1:N}$ consisting of $N$ samples from a multinomial distribution with probabilities proportional to $w_1,\ldots, w_N$. 
The algorithm is summarized in \autoref{alg:particle_filter}.

\begin{algorithm}[ht]
\caption{Particle filter with $N$ particles}
\label{alg:particle_filter}
\begin{algorithmic}[1]
\Statex{\emph{At time} $t=1$}
\State{Sample 
$\widetilde{X}_1^i\sim q(\cdot\mid y_1), \quad i=1\ldots N$
}
\State{Compute weights}
\begin{equation*}
	w_{1,i} = w_t(\widetilde{X}_1^i), \quad i = 1,\ldots, N
\end{equation*}
\State{$I_{1:N} \sim \textsf{Mult}\left\{N; w_{1,1}, \ldots, w_{1,N}\right\}$}
\State{Set $X_1^i = \widetilde{X}_1^{I_i}, \quad i=1.\ldots, N$}
\Statex{\emph{At times} $t\geq 2$}
\State{Sample
$\widetilde{X}_t^i \sim q(\,\cdot \mid  X_{t-1}^i, y_t), \quad i = 1,\ldots, N$
}
\State{Compute weights as in \eqref{eq:bf_weights}}
\begin{equation*}
	w_{t,i} = w_t(X_{t-1}^i, \widetilde{X}_t^i), \quad i = 1,\ldots, N
\end{equation*}
\State{$I_{1:N} \sim \textsf{Mult}\{N; w_{t,1}, \ldots, w_{t,N}\}$}
\State{For $i =1,\ldots, N$ set
\begin{align*}
X_{1:t}^i & = \left(X_{1:(t-1)}^{I_i},\widetilde{X}_{t}^{I_i}\right).
\end{align*}
}
\end{algorithmic}
\end{algorithm}
In most applications, interest lies not in the distribution itself,
but rather in the expectations
\begin{equation}\label{eq:Expectation}
	\mathcal{I}(h) = \int h(x_{1:t}) \pi_{t}(x_{1:t})dx_{1:t}
\end{equation}
for some test function $h\colon\mathsf{X}^{t}\rightarrow\mathbb{R}$. 
Applying a PF we can estimate this integral,
for a set of particle genealogies $\left\{ X_{1:t}^{i},i=1,\ldots,N\right\} $,
by taking
\begin{equation}\label{eq:smc_est}
\widehat{\mathcal{I}}_{t,\mathrm{PF}}(h)=\frac{1}{N}\sum_{i=1}^{N}h\left(X_{1:t}^i\right).
\end{equation}
When the PF weights \eqref{eq:bf_weights} are not available analytically,
standard resampling routines---steps 3 and 7 in \autoref{alg:particle_filter}---cannot be performed. However, in many cases one might be able to construct an unbiased estimator of the resampling weights. \textcite{liu1998sequential,rousset2006discussion,DelMoral2007,fearnhead2008particle} show that using random but unbiased
non-negative weights still yields a valid algorithm. This can be easily seen by
considering a standard particle filter on an extended space. Yet,
this flexibility does not come without cost. Replacing the true
weights with a noisy estimate increases the variance for Monte
Carlo estimates of type \eqref{eq:smc_est}.

Here we introduce a new resampling scheme that allows for multinomial
resampling from the true intractable weights while just requiring access to non-negative unbiased estimates.
Thus, any PF estimate of type \eqref{eq:smc_est} will have the same variance as if the true weights
were known. This algorithm relies on an extension of recent work in
\textcite{dughmi2017bernoulli} where the authors consider an unrelated problem.

In the supplementary material we collect the proofs for Propositions 3, 4 and 5, Theorem 6 and additional simulation studies. Code to reproduce our results is available online\footnote{URL: \url{http://www.stats.ox.ac.uk/~schmon/}.}.

\section{Particle Filters with Intractable Weights}

Particle filters for state space models with intractable weights
rely on the observation that replacing the true weights
with a non-negative unbiased estimate is equivalent to a particle filter on an
extended space. In this section we introduce some examples of models
where the weights are indeed not available analytically and review
briefly how the random weight particle filter (RWPF) can be applied in these instances.

\subsection{Locally optimal proposal\label{subsec:loc_opt}}

Recall that in the setting of a PF for a state space model at
time $t-1$ the proposal for $t$ which leads to the minimum one step
variance, termed the \emph{locally optimal proposal} $q^{*}$, is
\[
q^{*}(x_{t}\mid x_{t-1},y_{t})=\frac{g(y_{t}\mid x_{t})f(x_{t}\mid x_{t-1})}{p(y_{t}\mid x_{t-1})},
\]
see, e.g. \textcite[Proposition 2]{doucet2000sequential}. Sampling from
the locally optimal proposal is usually straightforward using a rejection sampler. The weights
\begin{align}\notag
w_{t}(x_{t-1},x_{t}) & =p(y_{t}\mid x_{t-1}) \\\label{eq:local_opt_weight}
										 & =\int g(y_{t}\mid x_{t})f(x_{t}\mid x_{t-1})dx_{t},
\end{align}
however, are intractable if the integral on the right-hand side does not
have an analytical expression, thus prohibiting an exact implementation
of this algorithm. Our algorithm will enable us to resample according to these weights.

\subsection{Partially observed diffusions}

Most research on particle filters with intractable weights has been
carried out in the setting of partially observed diffusions, see
e.g. \textcite{fearnhead2008particle}, which we will describe here briefly.
For simplicity, consider the univariate diffusion\footnote{It is not necessary to restrict oneself to univariate diffusions,
see e.g. \textcite{fearnhead2008particle}. } process of the form
\begin{equation}
\label{eq:const_diff}
dX_{t}=a(X_{t})dt+dB_{t},\quad t\in[0,\mathcal{T}],
\end{equation}
where $a\colon\mathbb{R}\rightarrow\mathbb{R}$ denotes the drift
function and $\mathcal{T} \in \mathbb{R}$ is the time horizon. For a general diffusion constant
speed as assumed in (\ref{eq:const_diff}) can be obtained whenever the Lamperti transform is available. The diffusion is observed at discrete
times $t_i, i=1, \ldots T$ with measurement error described by the
density $g(y_{t_i} \mid x_{t_i})$.
In this model the resampling weights are often intractable since for most diffusion processes the transition
density $f_{\Delta t}(x_{t_i}\mid x_{t_{i-1}})$ over the interval with length
$\Delta t = t_i - t_{i-1}$ is not available analytically.
However, as shown in \textcite{rogers1985smooth,dacunha1986estimation,beskos2005exact}
the transition density over an interval of length $\Delta t$ can be expressed as
\begin{align} \notag
f_{\Delta t}(y\mid x) & = \varphi(y;x,\Delta t)\exp\left(A(y)-A(x)\right) \\\label{eq:florens-zmirou}
& \quad \times \mathbb{E}\left[\exp\left(-\int_{0}^{\Delta t}\phi(W_{s})ds\right)\right]
\end{align}
where $\varphi(\,\cdot\,; x, \Delta t)$ is the density of a Normal distribution with mean $x$ and variance $\Delta t$, $\left(W_{s}\right)_{s\in[0,\Delta t]}$ denotes a Brownian bridge
from $x$ to $y$ and $\phi(x)=\left(a^{2}(x)+a'(x)\right)/2$
for a function $A\colon\mathbb{R}\rightarrow\mathbb{R}$ with $A'(x)=a(x)$.
The expectation on the right-hand side in (\ref{eq:florens-zmirou})
is with respect to the Brownian bridge and is usually intractable.
Thus, for a particle filter to work in this example one either needs
to implement a Bootstrap PF, which can sample the diffusion
(\ref{eq:const_diff}) exactly \autocite[see, e.g.][]{beskos2005exact},
or one constructs an unbiased estimator of the expectation on the
right-hand side to employ the RWPF. Note
that for a function $g$ and $U\sim\mathrm{Unif}[0,t]$ we have
\[
\mathbb{E}\left[\frac{g(W_{U})}{\lambda}\mid W_{s},0\leq s\leq t\right]=\int_{0}^{t}\frac{g(W_{s})}{\lambda t}ds.
\]
Using this relationship we can use debiasing schemes such as the Poisson
estimator \autocite{Beskos2006} to find a non-negative unbiased estimator of the expectation of the
exponential.

\subsection{Random weight particle filter}

Here we briefly review the RWPF and compare its asymptotic variance to that of a PF with known weights.
Assume one does not have access to the weight $\varpi\left(x_{t}\mid y_{t},x_{t-1}\right)$ in \eqref{eq:bf_weights} but only
to some non-negative unbiased estimate $\hat{\varpi}(x_{t}\mid y_{t},x_{t-1},U_{t})$,
where $U_{t}$ are auxiliary variables sampled from some density $m$.
Then we carry out the multinomial resampling in \autoref{alg:particle_filter} by taking
$I_{1:N} \sim \textsf{Mult}\{N; \hat{w}_{t,1}, \ldots, \hat{w}_{t,N}\}$, where $\hat{w}_{t, k}$
is defined as in \eqref{eq:bf_weights} with $\varpi\left(x_{t}\mid y_{t},x_{t-1}\right)$ replaced by $\hat{\varpi}(x_{t}\mid y_{t},x_{t-1},U_{t})$.
This is equivalent to a standard particle filter on the extended space targeting at time $t$ the distribution
\[
\bar{\pi}_{t}\left(x_{1:t},u_{1:t}\right)\propto\prod_{k=1}^{t}\hat{\varpi}(x_{k}\mid y_{k},x_{k-1},u_{k})m\left(u_{k}\right)
\]
which satisfies
\begin{equation}
\bar{\pi}_{t}\left(x_{1:t},u_{1:t}\right)=\pi_{t}(x_{1:t})\bar{\pi}_{t}(u_{1:t}\mid x_{1:t})\label{eq:marginal_decomp}
\end{equation}
with
\begin{equation*}
\bar{\pi}_{t}(u_{1:t}\mid x_{1:t}) =\prod_{k=1}^{t}\frac{\hat{\varpi}(x_{k}\mid y_{k},x_{k-1},u_{k})m\left(u_{k}\right)}{\varpi\left(x_{k}\mid y_{k},x_{k-1}\right)}.
\end{equation*}

A PF targeting the sequence of distributions $\pi_{t}$ directly can be interpreted
as a Rao-Blackwellized PF \autocite{doucet2000sequential} of the PF on the extended space. Hence,
the following is a direct consequence of \autocite[][Theorem 3]{chopin2004central}.
\begin{prop}
\label{prop:smaller_variance}
For any sufficiently regular\footnote{We refer the reader to \textcite{chopin2004central} for the mathematical details.} real-valued test function $h$, the exact weight PF (EWPF)
and RWPF estimators of $\mathcal{I}(h)$ defined in (\ref{eq:Expectation}) both satisfy a $\sqrt N$-central limit theorem with asymptotic 
variances satisfying $\sigma_{\text{EWPF},h}^{2}\leq \sigma_{\text{RWPF},h}^{2}$.
\end{prop}

\section{Bernoulli Races}

Assume now that the intractable weights can be written as follows
\begin{align} \notag
 & w_t(x_{t-1},x_{t})=\frac{g(y_{t}\mid x_{t})f\left(x_{t}\mid x_{t-1}\right)}{q\left(x_{t}\mid x_{t-1},  y_{t}\right)} \\ & =c\left(x_{t-1},x_{t}, y_t\right)b\left(x_{t-1},x_{t}, y_t\right),\label{eq:tract_replace}
\end{align}
where $0\leq b(x, x',y)\leq1$ for all $x,x', y\in\mathsf{X}$, and that we are
able to generate a coin flip $Z$ with $\mathbb{P}(Z=1)=b\left(x,x',y\right)$.
We will assume that $c(x,x',y)$ in (\ref{eq:tract_replace}) is analytically
available for any $x,x',y$. For brevity we will denote (\ref{eq:tract_replace})
as $w_{i}=c_{i}b_{i}$, for particles $i=1, \ldots,N$, dropping for now the dependence
on $t$.

The aim of this section is to develop an algorithm to perform multinomial sampling proportional
to the weights $w_i$, that is, sample from discrete
distributions of the form
\begin{equation}
p(i)=\frac{c_{i}b_{i}}{\sum_{k=1}^{N}c_{k}b_{k}},\quad i=1,\ldots N \label{eq:multinomial}
\end{equation}
where $c_{1},\ldots,c_{N}$ denote fixed constants whereas the $b_{1},\ldots,b_{N}$
are unknown probabilities. We assume we are able to sample coins
$Z_i\sim\mathrm{Ber}(b_i),i=1,\ldots,N.$ In practice, the $c_1, \ldots, c_N$ are selected to
ensure that $b_1, \ldots, b_N$ take values between 0 and 1. The Bernoulli race algorithm
for multinomial sampling proceeds by first proposing from the distribution
\begin{equation}
\mathbb{P}\left(I=i\right)=\frac{c_{i}}{\sum_{k=1}^{N}c_{k}}\label{eq:multi_prop}
\end{equation}
and then sampling $Z_{I}\sim \mathrm{Ber}(b_{I})$. If $Z_{{I}}=1$, return ${I}$,
otherwise the algorithm restarts. Pseudo code describing this procedure
is presented in Algorithm \ref{Ref:alg:bern_race}. If
we take $c_{j}=1$ for all $j=1,\ldots,N$ we recover the original
Bernoulli race algorithm from \textcite{dughmi2017bernoulli}.

\begin{algorithm}
\begin{algorithmic}[1]
\caption{Bernoulli race}
\label{Ref:alg:bern_race}
\State{Draw $I\sim \mathsf{Mult}\{1; c_1, \ldots, c_N\}$}
\State{Draw $Z \sim b_{{I}}$} \If {$Z=1$} 	
\Return{$I = I$}
\Else{ go back to line 1}
\EndIf
\end{algorithmic}
\end{algorithm}

\begin{prop}\label{prop:2}
Algorithm \ref{Ref:alg:bern_race} samples from the distribution
\[
p(i)=\mathbb{P}(I = i \mid Z_I = 1) = \frac{c_{i}b_{i}}{\sum_{k=1}^{N}c_{k}b_{k}}.
\]
\end{prop}
\begin{proof}

Note that the probability of sampling $I=i$ and accepting is
\begin{equation*}
	\mathbb{P}(I = i, Z_i = 1) = \frac{b_i c_i}{\sum_k c_k}.
\end{equation*}
It follows that observing $Z_I = 1$ has probability $\mathbb{P}(Z_I=1) = \sum_k b_k c_k / \sum_k c_k$. Now for any $i=1,\ldots N$
\begin{align*}
p(i) & = \mathbb{P}(I = i \mid Z_I = 1) \\
	& =\frac{\mathbb{P}(I = i, Z_i = 1)}{\mathbb{P}(Z_I = 1)} \\
  & =\frac{b_{i}c_{i}}{\sum_{k}c_{k}}\Big/\frac{\sum_{k}b_{k}c_{k}}{\sum_{k}c_{k}} =\frac{b_{i}c_{i}}{\sum_{k}b_{k}c_{k}}.\qedhere
\end{align*}
\end{proof}

\subsection{Efficient Implementation}

This algorithm repeatedly proposes an a priori unknown amount of random
variables from the multinomial distribution (\ref{eq:multi_prop}).
Naïve implementations of multinomial sampling are of complexity $O(N)$
for one draw from the multinomial distribution. Standard resampling
algorithms, however, can sample $N$ random variables at cost $O(N)$ \autocite[see e.g.][]{hol2006resampling}
but in this case the number of samples needs to be known beforehand.
We show here that Bernoulli resampling can still be implemented with
an average of $O(N)$ operations.

To ensure that the Bernoulli resampling is fast we need cheap samples from the multinomial distribution with weights proportional to $c_1,\ldots,c_N$. We can achieve this
with the Alias method \autocite{walker1974new,walker1977efficient} which
requires $O(N)$ preprocessing after which we can sample with $O(1)$
complexity. Hence, the overall complexity depends on the number of
calls to the Bernoulli factory per sample. Denote $C_{j,N}$ the number
of coin flips that are required to accept a value for the $j$th sample, where $j=1,\ldots,N$. The random variable $C_{j,N}$ follows a geometric distribution with success probability
\begin{equation}
\rho_{N}= \mathbb{P}(Z_I = 1) = \frac{\sum_{k=1}^{N}c_{k}b_{k}}{\sum_{k=1}^{N}c_{k}}.\label{eq:geometric_success}
\end{equation}
The expected number of trials until a value is accepted is then
\[
\mathbb{E}\left[C_{j,N}\right]=\frac{1}{\rho_N}\quad\left(j=1,\ldots,N\right).
\]

The complexity of the resampling algorithm depends on the values of
the acceptance probabilities $b_{1},\ldots,b_{N}$. For example, if
all probabilities are identical, that is, $b_{1}=\ldots=b_{N}=b$, we
expect $N/b$ coin flips, each of cost $O(1)$, which
leads to overall order $O(N)$ complexity. In practice, however, the
success probabilities of our Bernoulli factories are all different.
Assuming all $b_{j}, j=1, \ldots, N$ are non-zero, the expected algorithmic complexity
per particle is bounded by the inverse of smallest and largest Bernoulli
probabilities. Let $\underline{b}=\min\{b_{1},\ldots,b_{N}\}$.
Then,
\[
\mathbb{E}[C_{j,N}]=\frac{\sum_{k}c_{k}}{\sum_{k}b_{k}c_{k}}\leq\frac{\sum_{k}c_{k}}{\underline{b}\sum_{k}c_{k}}=\frac{1}{\underline{b}}
\]
and with $\overline{b}=\max\{b_{1},\ldots,b_{m}\}$
\[
\mathbb{E}[C_{j,N}]\geq\frac{1}{\overline{b}}.
\]
In practice, the complexity of the Bernoulli sampling algorithm depends
on the behavior of $\rho_{N}$ as shown by the following central limit theorem.
\begin{prop}
\label{prop:clt1}Assume $\lim_{N\rightarrow\infty}\rho_{N}=:\rho\in(0,1)$.
Then we have the following central limit theorem for the average number
of coin flips as $N\rightarrow\infty$
\[
\sqrt{N}\left(\frac{1}{N}\sum_{j=1}^{N}C_{j,N}-\frac{1}{\rho_{N}}\right)\overset{d}{\rightarrow}\mathcal{N}\left(0,\frac{1-\rho}{\rho^{2}}\right),
\]
where $\overset{d}{\rightarrow}$ denotes convergence in distribution.
\end{prop}

In particular, Proposition \ref{prop:clt1} implies that the run-time
concentrates around $N/\rho_{N}$ with fluctuations of order $\sqrt{N}$.
Thus, the order of complexity depends on the order of $\rho_{N}$
and we have complexity $O(N)$ if
\begin{equation}
\limsup_{N\rightarrow\infty}\left|\frac{\sum_{k=1}^{N}c_{k}}{\sum_{k=1}^{N}c_{k}b_{k}}\right|<\infty.\label{eq:order_N_cond}
\end{equation}

As an example consider the case where $c_{1}=\ldots=c_{N}=c$ are all identical. Then at time $t$
\begin{equation}
\rho_{t, N}=\frac{1}{N}\sum_{k=1}^{N}w_{t, k}.\label{eq:norm_est}
\end{equation}
An instance of this setting is the locally optimal proposal presented
in Section \ref{subsec:loc_opt}. It is well known that as $N\rightarrow\infty$
(\ref{eq:norm_est}) will converge towards $p(y_{t}\mid y_{1:t-1})$
and indeed we see that the algorithm's run-time will concentrate around
a quantity of order $N$.

\begin{rem}
After the Alias table is constructed, the algorithm can be implemented in
parallel. This can lead
to considerable gains if the number of particles used in the particle
filter is high.
If $c_{1}=\ldots=c_N=c$, or if the constants denote a distribution for which a table as in the Alias method is already implemented before program execution, the above algorithm can be implemented entirely in parallel.
\textcite{murray2016parallel} consider the case of multinomial
resampling using a rejection sampler with uniform proposal
on the set $\{1,\ldots,N\}$ with known weights
and show that this algorithm has parallel complexity $O(\log N).$
\end{rem}

\subsection{Estimating the Probability of Stopping}

In our later applications we will be interested in evaluating the
success probability
\[
\rho_{N}=\frac{\sum_{k=1}^{N}c_{k}b_{k}}{\sum_{k=1}^{N}c_{k}}.
\]
Assume that we sample $N$ independent realizations from a multinomial
distribution using the Bernoulli race algorithm described above and that
$C_{1,N},\ldots,C_{N,N}$ are the geometric random variables that
count the number of trials until the algorithm accepts a value and
terminates. Then $\hat{\rho}_N^{\mathrm{naive}}=1/\bar{C}_{N}$, where
$\bar{C}_{N}=\sum_{k=1}^{N}C_{k,N}/N$ is a consistent estimator of
$\rho$ since $\mathbb{E}(C_{i,N})=1/\rho_N$ for all $i=1,\ldots,N$
and therefore by a weak law of large numbers $\bar{C}_{N}{\rightarrow}1/\rho$ in probability
and $1/\bar{C}_{N}{\rightarrow}\rho$ in probability by the continuous
mapping theorem. Unfortunately, the estimator $\hat{\rho}_N^{\mathrm{naive}}$
is not unbiased which would be desirable. However, this can be remedied
by constructing the minimum variance unbiased estimator \autocite[see][Definition 7.1.1]{hogg2005introduction} for a geometric distribution. This result is well known, but we repeat it here for
convenience.
\begin{prop}
For $N\in\mathbb{N}$ let $C_{1,N},\ldots,C_{N,N}$ denote independent
samples from a geometric distribution with success probability $\rho_N$,
then the minimum variance unbiased estimator for $\rho_N$ is
\begin{equation}
\hat{\rho}_{N}^{\mathrm{mvue}}=\frac{N-1}{\sum_{k=1}^{N}C_{k,N}-1}.\label{eq:geom_umvue}
\end{equation}
\end{prop}

In order to understand the asymptotic behavior of the estimator (\ref{eq:geom_umvue})
we also provide the following central limit theorem.
\begin{prop}
For $N\in\mathbb{N}$ let $C_{1,N},C_{2,N},\ldots$ denote independent
samples from a geometric distribution with success probability $\rho_N$,
then
\[
\sqrt{N}\left(\hat{\rho}_{N}^{\mathrm{mvue}}-\rho_{N}\right)\overset{d}{\rightarrow}\mathcal{N}\left(0,\left(1-\rho\right)\rho^{2}\right).
\]
\end{prop}

\section{Bernoulli Race Particle Filter}

\subsection{Algorithm Description}

We now consider the application of the Bernoulli race algorithm
to particle filter methods. The Bernoulli race can be employed
as a multinomial resampling algorithm, \textsf{Mult}$\{w_1, \ldots, w_N\}$, in \autoref{alg:particle_filter}.
However, the clear advantage is that this algorithm can be implemented even if the true
weights are not analytically available, but we do have
access to non-negative unbiased estimates. A $[0,1]$-valued unbiased estimator $\hat{b}$ for $b$ can be converted into an unbiased
coin flip by noting that $\mathbb{P}\left(V \leq \hat{b}\right) = b,$ where $V \sim \mathrm{Unif}[0,1]$ \autocite[see e.g. Lemma 2.1 in][]{latuszynski2011simulating}.

We will refer to such a particle filter as a Bernoulli race particle filter (BRPF).

\subsection{Likelihood estimation\label{subsec:likelihood_est}}

Even though the Bernoulli race resampling scheme enables us to resample according to
the true weights, the normalizing constant or marginal likelihood \eqref{eq:marginal_likelihood} remains intractable.
We show here how we can still obtain an unbiased
estimator for the normalizing constant. We first recall that in particle filters an estimator of the normalizing constant is obtained by
\begin{equation}
\hat{p}(y_{1:T})=\prod_{t=1}^{T}\hat{p}(y_{t}\mid y_{1:t-1})=\prod_{t=1}^{T}\frac{1}{N}\sum_{k=1}^{N}w_{t,k}.\label{eq:unbiased_est}
\end{equation}
This estimator is well-known to be unbiased \autocite[see][Chapter 7]{del2004feynman}. If the weights are not available this estimator cannot be employed. Fortunately, the quantity (\ref{eq:unbiased_est})
comes up naturally when running the BRPF. Recall
that the probability for the Bernoulli race to stop at a given iteration,
i.e. to accept a value, is
\[
\mathbb{P}\left(C_{j,N} = 1\right) = \frac{\sum_{k=1}^{N}c_{t, k}b_{t, k}}{\sum_{k=1}^{N}c_{t, k}}=\frac{\frac{1}{N}\sum_{k=1}^{N}w_{t, k}}{\frac{1}{N}\sum_{k=1}^{N}c_{t, k}}.
\]
Thus, conditional on the weights $w_{t,k}, k=1,\ldots,N$,
an unbiased estimator of (\ref{eq:unbiased_est}) is given by virtue
of (\ref{eq:geom_umvue}):
\begin{equation}
\label{eq:unbiased_est_bernoulli}
\hat{\rho}_{N,T}=\prod_{t=1}^{T}\frac{1}{N}\sum_{k=1}^{N}c_{t, k}\cdot\frac{N-1}{\sum_{k=1}^{N}C_{k,N}-1}.
\end{equation}

\begin{thm}
The estimator \eqref{eq:unbiased_est_bernoulli} is unbiased for $p(y_{1:T})$, i.e. $\mathbb{E}\left[\hat{\rho}_{N,T}\right] = p(y_{1:T})$.
\end{thm}

\section{Applications}

\subsection{Locally optimal proposal}

We start with a Gaussian state space model. In linear Gaussian state
space models the Kalman Filter can be used to analytically evolve
the system through time. Nevertheless, the aim here is a proof of concept of the BRPF. We assume the latent variables follow
the Markov chain
\[
X_t=aX_{t-1}+V_{t}
\]
where we take $a = 0.8$ and we observe these hidden variables through the observation equation
\[
Y_{t}=X_{t}+W_{t}
\]
with initialization $X_{0}\sim\mathcal{N}(0,5)$ and $V_{t}\sim\mathcal{N}(0,5)$,
$W_{t}\sim\mathcal{N}(0,5)$. In this particular instance the locally optimal proposal is available analytically, but in most practical scenarios this will not be the case. For this reason  sampling from the locally optimal proposal is implemented using a rejection sampler that proposes from the state equation.  
Coin flips for the weights \eqref{eq:local_opt_weight} can be obtained by sampling from the model \mbox{$\xi_t \sim f(\cdot \mid x_{t-1})$} and computing 
\begin{equation*}
	Z_t = 1 \left\{U \leq \exp\left(-\frac{(y_t - \xi_t)^2}{10} \right)\right\}, \quad U\sim \mathrm{Unif}[0, 1].
\end{equation*}
This leads to the choice $c_{t, 1} = \ldots = c_{t, N} = 1/\sqrt{10\pi}$ and 
\begin{equation*}
	b_{t,k} = \int \exp\left(-\frac{(y_t - x)^2}{10} \right) f(x \mid x_{t-1}) dx_{t-1}.
\end{equation*}
Note that $c_{t,k}$ is defined such that 
\begin{align*}
	\sup_{x, x', y} b_t(x, x', y) = 1
\end{align*}
Such a choice ensures that the acceptance probability in the Bernoulli race, \autoref{Ref:alg:bern_race}, is as large as possible. This can lead to a considerable speedup when using the Bernoulli resampling algorithm.

\subsubsection{Complexity and Run-time}

Figure \ref{fig:kf_comp} shows the run-time for RWPF and the BRPF for the Gaussian state space model.
The BRPF clearly has linear complexity in the number of particles. In a sequential implementation
the Bernoulli race (orange) performs worse than the classical resampling scheme (blue), but the difference vanishes when implementing a parallel version of the Bernoulli race (green, with 32 cores). This demonstrates the speedup due to parallel sampling of the coin flips. Since we need many Bernoulli random variables each of which is computationally cheap this algorithm lends itself to a implementation using GPUs to further improve the performance of the Bernoulli race.

\begin{figure}[h]
\includegraphics[width=1\linewidth]{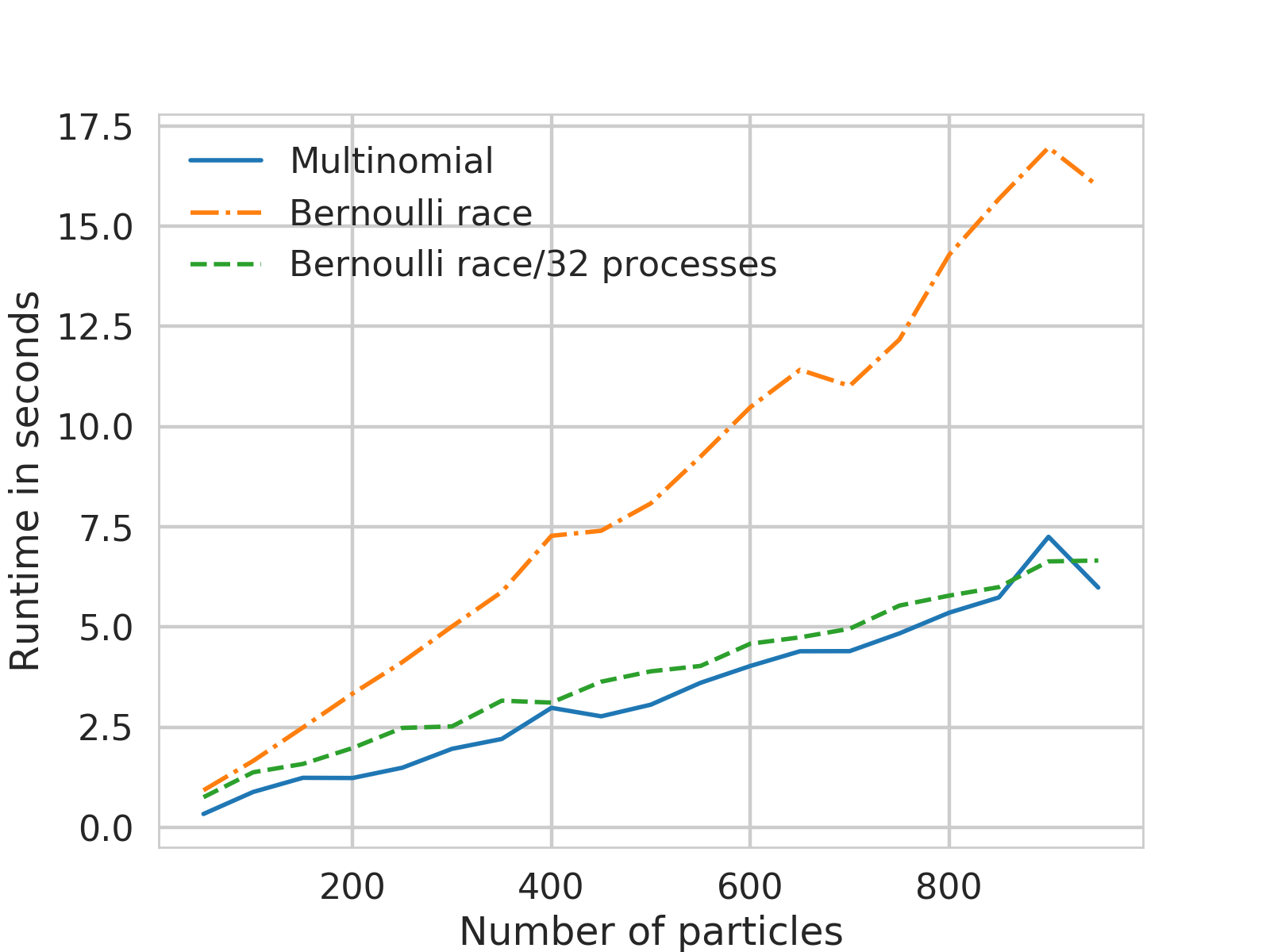}
\caption{Comparison of the run-time for RWPF and BRPF for Gaussian state space model. For BRPF we show a sequential
and a parallel implementation with 32 cores. For all algorithms we show wall-clock time
for number of particles $N$. }
\label{fig:kf_comp}
\end{figure}

\subsubsection{Efficiency}

As alluded to earlier, if Bernoulli resampling is performed,
the variance for any Monte Carlo estimate (\ref{eq:smc_est}) will
be the same as if the true weights were known and one applies
standard multinomial resampling. From Proposition \ref{prop:smaller_variance} it follows that the asymptotic variance of any Monte Carlo estimate
of type (\ref{eq:smc_est}) will be smaller when applying a BRPF over a RWPF. While the variance for functions (\ref{eq:smc_est}) in the BRPF coincides with the standard particle filter we do not have access to the same estimator for the normalising constant and instead need to use the methods from Section \ref{subsec:likelihood_est}.

For these reasons we will compare the performance of the BRPF with the RWPF for a set of different test functions $h\colon \mathsf{X}^T \rightarrow \mathbb{R}$ by comparing the variance of PF estimates \eqref{eq:smc_est}.
We then study the estimation of the normalizing constant separately. For the simulation
we consider the functions
\begin{align*}
h_{1}(x_{1:T}) & =\frac{1}{T}\sum_{t=1}^{T}x_{t}, &\quad h_{2}(x_{1:T}) & =\|x_{1:T}\|_{2},\\
h_{3}(x_{1:T}) & =x_{T}, &\quad h_{4}(x_{1:T}) & =\left(x_{T}-\bar{x}_{T}\right)^{2},
\end{align*}
where $\bar{x}_{T}=\sum_{i=1}^{N}x_{T}^{i}/N$. The PF approximations \eqref{eq:smc_est} using $h_{3}$ and $h_{4}$ are estimators for the quantities \mbox{$\mu_{T}=\int x_{T}p(dx_{T}\mid y_{1:T})$}
and \mbox{$\int\left(x_{T}-\mu_{T}\right)^{2}p(dx_{T}\mid y_{1:T})$}.

All estimates are based on 100 runs of each particle filter using $N=100$ particles. The
results are collected in Table \ref{tab:gaussian_ssm}. We denote
the standard deviation of the test functions under the BRPF as $\sigma_{\mathrm{BRPF}}$ and $\sigma_{\mathrm{RWPF}}$ for the RWPF. 
All measures indicate a reduction in the standard deviation. For the estimator of the function $h_{1}$, the standard deviation is
reduced by $26\%$ when compared to the RWPF.
\begin{table}
\caption{For test functions $h_{i},\ldots i=1,\ensuremath{\ldots,4}$ the
standard deviation for RWPF and BRPF over 100 iterations.}
\label{tab:gaussian_ssm}
\begin{center}%
\begin{tabular}{cccc}
\toprule
Test function & $\sigma_{\mathrm{RWPF}}$ & $\sigma_{\mathrm{BRPF}}$ & $\sigma_{\mathrm{BRPF}}/\sigma_{\mathrm{RWPF}}$\tabularnewline
\midrule
\midrule
$h_{1}$ & 0.17 & 0.12 & $0.74$\tabularnewline
\midrule
$h_{2}$ & 0.93 & 0.78 & $0.84$\tabularnewline
\midrule
$h_{3}$ & 0.19 & 0.19 & $0.96$\tabularnewline
\midrule
$h_{4}$ & 0.42 & 0.40 & $0.94$\tabularnewline
\bottomrule
\end{tabular}
\end{center}
\end{table}

We now investigate the estimates for the normalizing constant. \autoref{tab:gaussian_ssm_Z}
shows the standard deviation of the \mbox{(log-)}normalizing constant of a Gaussian
state space model for $T=50$ time steps. In this setting it is not
obvious why the Bernoulli race resampling estimate should outperform the
estimate provided by the RWPF. In our case however, we
find that the BRPF performs better.

\begin{table}
\caption{Comparison of the normalizing constant estimate for different implementations
of the particle filter. }
\label{tab:gaussian_ssm_Z}
\begin{center}%
\begin{tabular}{cc}
\hline
 & $\mathrm{sd}(\log\hat{p}(y_{1:T}))$\tabularnewline
\hline
\hline
BRPF & 0.55\tabularnewline
\hline
RWPF & 0.66\tabularnewline
\hline
\end{tabular}
\end{center}
\end{table}

\subsection{Partially Observed Diffusion}

We use the sine diffusion, a commonly used example in the context of partially observed diffusions, see e.g. \textcite{fearnhead2008particle},
given by the stochastic differential equation (SDE)
\begin{equation*}
dX_{s}=\sin(X_{s})dt+dB_{s}, \quad s \in [0, 15].
\end{equation*}
Here, $(B_{s})_{s\in [0, 15]}$ denotes a Brownian motion and the drift function in
(\ref{eq:const_diff}) is $a(x)=\sin(x).$ Consequently, $A(x)=-\cos(x)$
and with $\phi(x)=\left(\sin(x)^{2}+\cos(x)\right)/2$, the transition
density is
\begin{align}\notag
f_{\Delta t}(x,y) &= \varphi(y;x,\Delta t)\exp\left(-\cos(y)+\cos(x)\right)\\\label{eq:sinx}
& \quad \times \mathbb{E}\left[\exp\left(-\int_{0}^{\Delta t}\phi(W_{s})ds\right)\right].
\end{align}
We observe the state of the SDE through zero mean Gaussian noise with standard deviation 5 yielding weights
\begin{equation*}
	w(x_{t-1}, x_t, y_t) = \frac{\varphi(y_{t}; x_{t}, 5^2)f_{\Delta t}(x_{t} \mid x_{t-1})}{q(x_{t}\mid x_{t-1}, y_t)},
\end{equation*}
As the proposal $q(\cdot \mid x_{t-1}, y_t)$ we take one step of the Euler--Maruyama scheme.

\begin{figure}[h]
\includegraphics[width=\linewidth]{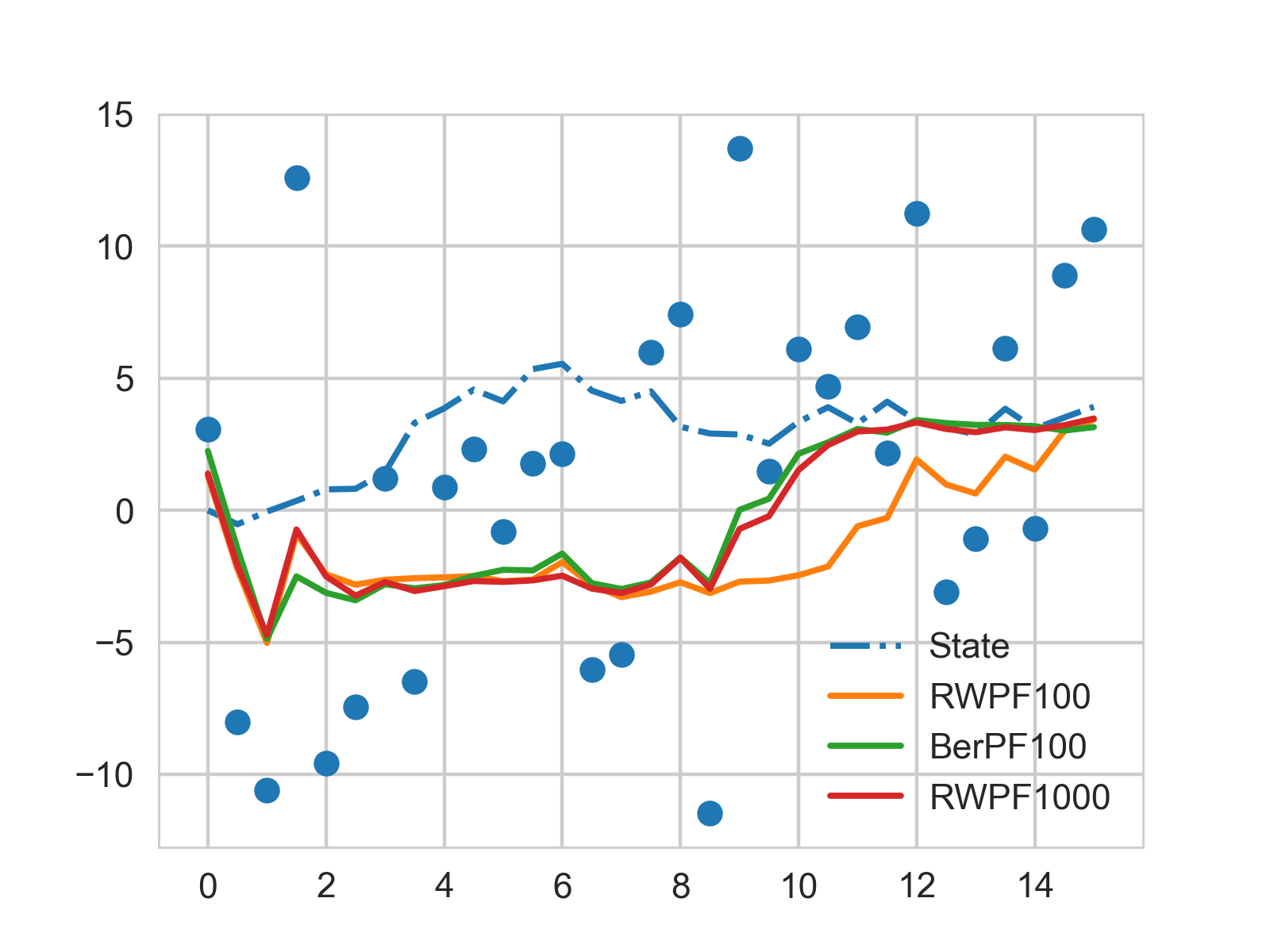}
\caption{Tracking performance for different particle filters; RWPF with 100
particles (RWPF100) and 1000 particles (RWPF1000) as well as a particle filter using Bernoulli resampling
with 100 particles (BerPF100).}
\label{fig:tracking}
\end{figure}

The weights are intractable because of the expectation on the right-hand side in \eqref{eq:sinx}.
The construction of unbiased coin flips can often rely on similar techniques as the construction of unbiased estimators.
We can construct an unbiased estimator for the transition density in the following way. Fix a Brownian bridge $\left(W_{s}\right)_{s\in[0,\Delta t]}$
starting at $x$ and finishing at $y$. Sample $\kappa\sim\mathrm{Pois}(\lambda \Delta t)$,
$U_{1}\ldots,U_{\kappa}\sim U[0,\Delta t]$ then the Poisson estimator \autocite{Beskos2006}
is given by
\[
\hat{P}(\kappa,U_{1}\ldots,U_{\kappa})=\exp\left\{ \left(\lambda-c\right)\Delta t\right\} \prod_{i=1}^{\kappa}\frac{\left\{ c-\phi\left(W_{U_{i}}\right)\right\} }{\lambda},
\]
where $c$ is chosen such that the estimator is non-negative. The Poisson estimator is unbiased
\begin{equation}
\mathbb{E}\left[\hat{P}(\kappa,U_{1}\ldots,U_{\kappa})\right]=\mathbb{E}\left[\exp\left(-\int_{0}^{\Delta t}\phi(W_{s})ds\right)\right].\label{eq:exp_bb}
\end{equation}

\begin{figure}
\includegraphics[width=\linewidth]{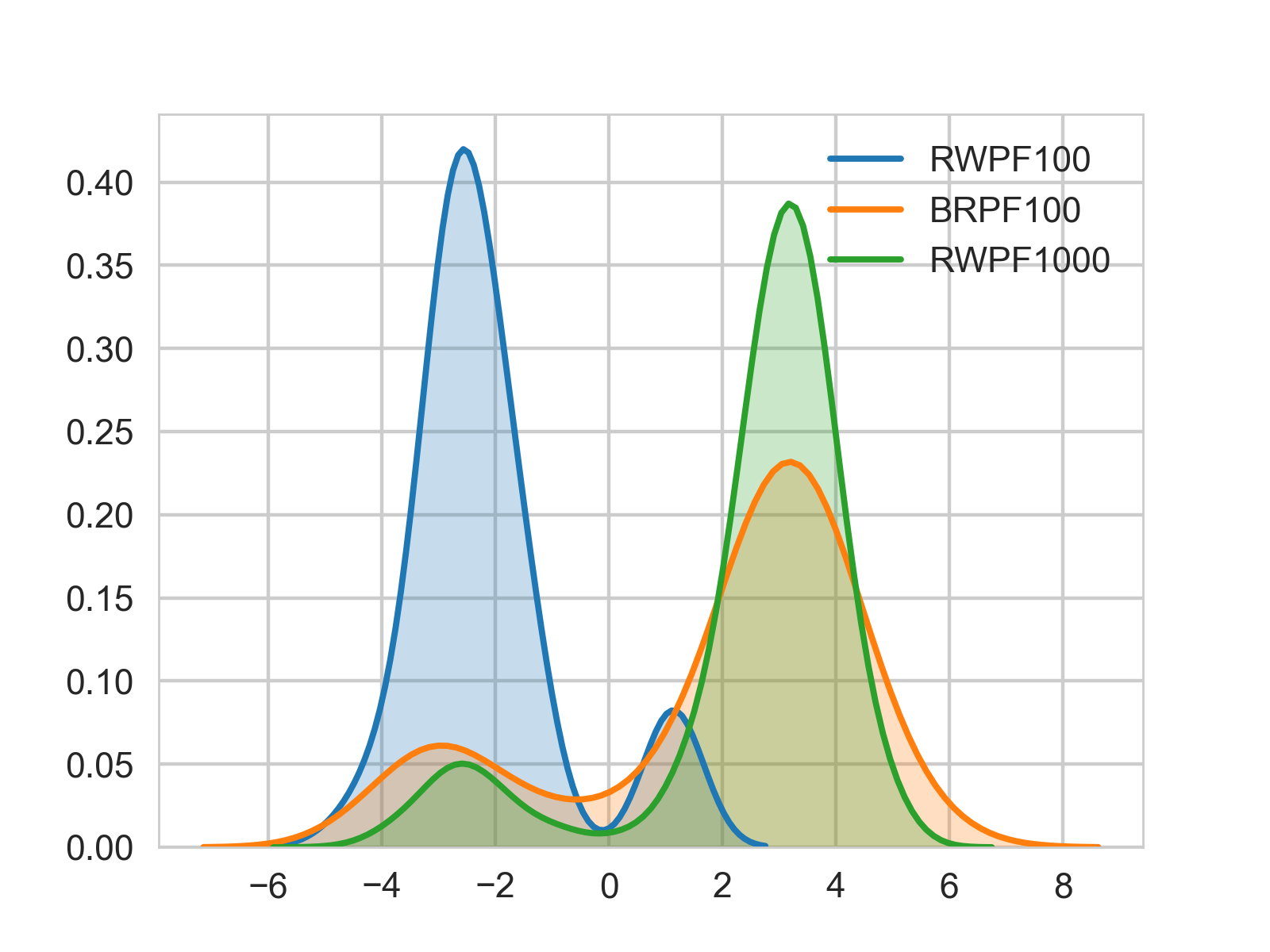}
\caption{Kernel density estimate of $p(x_{10}\mid y_{1:10})$ based on particle approximations from a RWPF with 100 particles (RWPF)
a BRPF with 100 particles (Bernoulli) and a RWPF with 1000 particles.}
\label{fig:kde_plot}
\end{figure}

For the purposes of implementing the BRPF, we need to construct
a coin flip with success probability proportional to (\ref{eq:exp_bb}). 
We use the probability
generating function approach which works analogously to the Poisson estimator.
Sample $\kappa \sim \mathrm{Pois}(\lambda\Delta t), V_i\sim \mathrm{Unif}[0,1], i=1, \ldots \kappa$,
then
\begin{equation*}
	Z = \prod_{i=1}^\kappa 1\left\{V_i \leq \frac{\left\{ c-\phi\left(W_{U_{i}}\right)\right\} }{\lambda}\right\}
\end{equation*}
has success probability
\begin{equation*}
	\exp\left\{(c-\lambda) \Delta t\right\}\mathbb{E}\left[\exp\left(-\int_{0}^{\Delta t}\phi(W_{s})ds\right)\right].
\end{equation*}
Hence, to implement the BRPF we can choose 
\begin{align*}
	c_{t,k} & = \frac{\varphi(y_t;x_t,5^2)\varphi(x_t;x_{t-1},\Delta t)}{\varphi(y_t;x_{t-1} +  \Delta t\sin(x_{t-1}),\Delta t)} \\
				& \quad \times \exp\left(-\cos(x_{t})+\cos(x_{t-1})-(c-\lambda) \Delta t\right)\\
	b_{t,k} & = \exp\left\{(c-\lambda) \Delta t\right\}\mathbb{E}\left[\exp\left(-\int_{0}^{\Delta t}\phi(W_{s}^k)ds\right)\right],
\end{align*}
where $(W_s^k)_{s\in[0, \Delta t]}$ denotes a Brownian bridge from $x_{t-1}^k$ to $x_{t}^k$. 

\subsubsection{Complexity and Run-time}

As in the previous example, we observe the BRPF to be of order $O(N)$ and slower than the RWPF. However, implementing
the Bernoulli race in parallel yields almost the same performance in terms of run-time. The details can be found in the 
supplementary material.

\subsubsection{Efficiency}

One run for the RWPF and the BRPF is shown in \autoref{fig:tracking}, where we show the true
state of the SDE as well as the noisy observations and the particle filter approximations. In this scenario
precise state estimation is hampered by the high noise and resulting partial multimodality of the filtering distribution as shown
in \autoref{fig:kde_plot}. For $N=100$, \autoref{fig:tracking} shows that the BRPF
performs much better as it finds the true state earlier. The RWPF finds this trajectory
only when the number of particles is increased (here we show also the case $N=1000$).
With the same test functions as above we compare both algorithms
in \autoref{tab:sine_ssm_test}. We observe gains for all functions, with the most significant gain for the conditional mean, $h_3$. Again, the Bernoulli race will
ordinarily be slower, but most of the difference in run-time vanishes when the Bernoulli race is implemented
in parallel. Further details are provided in the supplementary material.

\begin{table}
\caption{Comparison of the normalizing constant estimate for different implementations
of the particle filter. }
\label{tab:sine_ssm_test}
\begin{center}
\begin{tabular}{cccc}
\toprule
Test function & $\sigma_{\mathrm{RWBF}}$ & $\sigma_{\mathrm{BRPF}}$ & $\sigma_{\mathrm{BRPF}}/\sigma_{\mathrm{RWPF}}$\tabularnewline
\midrule
\midrule
$h_{1}$ & 1.41 & 1.27 & 0.91 \tabularnewline
\midrule
$h_{2}$ & 1.61 & 1.58 & 0.98 \tabularnewline
\midrule
$h_{3}$ & 1.26 & 0.64 & 0.50 \tabularnewline
\midrule
$h_{4}$ & 1.09 & 0.87 & 0.80 \tabularnewline
\bottomrule
\end{tabular}
\end{center}
\end{table}

\begin{table}
\caption{Comparison of the normalizing constant estimate for different implementations
of the particle filter. }
\label{tab:normal_constant}
\begin{center}%
\begin{tabular}{ccc}
\hline
 & $\mathrm{sd}(\log\hat{p}(y_{1:T}))$\tabularnewline
\hline
\hline
RWPF  & 3.83 \tabularnewline
\hline
BRPF & 3.11 \tabularnewline
\hline
Bootstrap & 3.13\tabularnewline
\hline
\end{tabular}
\end{center}
\end{table}
As a final test we use both particle filters for estimating the ($\log$-)normalizing
constant. The results are listed in \autoref{tab:normal_constant}. For comparison
we also employ a bootstrap particle filter using the exact algorithm \autocite{beskos2005exact}
to propose from the model. We observed this implementation to be slower than the other two.
The BRPF estimate \eqref{eq:unbiased_est_bernoulli} outperforms the RWPF and the bootstrap PF.

\section{Conclusion}

We have introduced the idea of Bernoulli races to resampling in particle
filtering, utilizing the equivalence of unbiased estimation and the construction
of unbiased coin-flips. This algorithm provides the first resampling
scheme to allow for exact implementation of multinomial resampling
when the weights are intractable, but unbiased estimates are available.
We have shown that our algorithm has a complexity of order $N$, like standard multinomial resampling, and demonstrated its advantages over alternative methods in a variety of settings. In doing so we have focused our attention to the resampling step. Further gains using
our algorithm could be obtained by considering auxiliary particle filters and to
resample only if the effective sample size drops beyond a certain threshold. We expect both
to positively impact the performance of the BRPF.

\printbibliography[heading=subbibliography] 
\end{refsection}

\appendix
\onecolumn
\hsize\textwidth
\linewidth\hsize \toptitlebar {\centering
{\Large\bfseries Supplementary Material to Bernoulli Race Particle Filters \par}}
\bottomtitlebar 

\setcounter{thm}{2}
\begin{refsection}

\section{Proofs}
When we say that $T$ is a geometric random variable with parameter, or success probability, $p$ we mean that 
$$\mathbb{P}(T\geq k) =(1-p)^k, \qquad k\geq 1.$$
\begin{prop}
Assume $\lim_{N\rightarrow\infty}\rho_{N}=:\rho\in(0,1)$.
Then we have the following central limit theorem for the average number
of coin flips as $N\rightarrow\infty$
\[
\sqrt{N}\left(\frac{1}{N}\sum_{j=1}^{N}C_{j,N}-\frac{1}{\rho_{N}}\right)\overset{d}{\rightarrow}\mathcal{N}\left(0,\frac{1-\rho}{\rho^{2}}\right),
\]
where $\overset{d}{\rightarrow}$ denotes convergence in distribution.
\end{prop}

 \begin{proof}
 The random variables $Z_{i,N}=\left(C_{i,N}-1/\rho_{N}\right)/\sqrt{N}$
 form a triangular array 
 \[
 \left\{ Z_{i,N};i=1,\ldots N,N\in\mathbb{N}\right\} 
 \]
  of row wise independent random variables. 
  Define $s_{N}^{2}=\sum_{i=1}^{N}\mathrm{var}\left(Z_{i,N}\right)=\left(1-\rho_{N}\right)/\rho_{N}^{2}$ and notice that $s_N^2 \to 1-\rho \in (0,1)$ as $N\to \infty$. 
In addition one can easily show that $\mathbb{E}[|Z_{1,N}|^3]$ are bounded uniformly in $N$ and therefore
\begin{align*}
\frac{1}{s_N^3} \sum_{i=1}^N \mathbb{E} \left[Z_{i,N} \right]
&= \frac{\rho_N^{3/2}}{(1-\rho_N)^3}\frac{N}{N^{3/2}}\mathbb{E}\left[\left|Z_{1,N} \right|^3\right]\\
&\leq \frac{C}{\sqrt{N}} \to 0, \qquad \text{as $N\to \infty$},
\end{align*}
that is the \textit{Lyapunov condition} is satisfied and therefore 
 $N\rightarrow\infty$
 \[
 \frac{1}{\sqrt{N}}\sum_{i=1}^{N}\left(C_{i,N}-\frac{1}{\rho_{N}}\right)\overset{d}{\rightarrow}\mathcal{N}\left(0,s^{2}\right),
 \]
 where $s^{2}=\lim_{N\rightarrow\infty}s_{N}^{2}=(1-\rho)/\rho^{2}.$
 \end{proof}

\begin{prop}
For $N\in\mathbb{N}$ let $C_{1,N},\ldots,C_{N,N}$ denote independent
samples from a geometric distribution with success probability $\rho_N$,
then the minimum variance unbiased estimator for $\rho_N$ is
\begin{equation}
\hat{\rho}_{N}^{\mathrm{mvue}}=\frac{N-1}{\sum_{k=1}^{N}C_{k,N}-1}.
\end{equation}
\end{prop}

\begin{proof}
This is a straightforward application of the Lehmann-Scheff\'e Theorem \autocite[Theorem 7.4.1]{hogg2005introduction}. Take the unbiased estimator $1\left\{C_{1,N} = 1\right\}$, where $1\{A\}$ denotes the indicator function of the set $A$, and condition on the complete and sufficient statistic (of the coin flips) $\sum_{k=1}^N C_{k,N}$. A straightforward calculation yields
\begin{align*}
	\mathbb{E}\left[1\left\{C_{1,N} = 1\right\} \mid \sum_{k=1}^N C_{k,N} = n \right] & = \frac{\mathbb{P}\left(C_1 = 1, \sum_{k=2}^N C_{k,N} = n-1 \right)}{\mathbb{P}\left(\sum_{k=1}^N C_{k,N} = n\right)} \\
	& = \frac{p\binom{n-2}{N-2}p^{N-1}(1-p)^{n - N}}{\binom{n-1}{N-1}p^{N}(1-p)^{n - N}} \\
	& = \frac{N-1}{n-1}
\end{align*} 
and thus 
\begin{equation*}
\hat{\rho}_{N}^{\mathrm{mvue}} = \mathbb{E}\left[1\left\{C_{1,N} = 1\right\} \mid \sum_{k=1}^N C_{k,N} \right] = \frac{N-1}{\sum_{k=1}^N C_{k,N}-1}. \qedhere
\end{equation*}
\end{proof}

\begin{prop}
For $N\in\mathbb{N}$ let $C_{1,N},C_{2,N},\ldots$ denote independent
samples from a Geometric distribution with success probability $\rho_N$,
then 
\[
\sqrt{N}\left(\hat{\rho}_{N}^{\mathrm{mvue}}-\rho_{N}\right)\overset{d}{\rightarrow}\mathcal{N}\left(0,\left(1-\rho\right)\rho^{2}\right).
\]
\end{prop}

 \begin{proof}
 By Proposition \ref{prop:clt1} we have convergence 
 \[
 \sqrt{N}\left(\frac{1}{N}\sum_{i=1}^{N}C_{i,N}-\frac{1}{\rho_{N}}\right)\overset{d}{\rightarrow}\mathcal{N}\left(0,\frac{1-\rho}{\rho^{2}}\right).
 \]
 In addition, 
 \begin{align*}
 \left|\sqrt{N}\left(\frac{\sum_{i=1}^{N}C_{i,N}-1}{N-1}-\frac{1}{\rho_{N}}\right)-\sqrt{N}\left(\frac{1}{N}\sum_{i=1}^{N}C_{i,N}-\frac{1}{\rho_{N}}\right)\right| & \rightarrow 0
 \end{align*}
 almost surely, so 
 \[
 \sqrt{N}\left(\frac{\sum_{i=1}^{N}C_{i,N}-1}{N-1}-\frac{1}{\rho_{N}}\right)\overset{d}{\rightarrow}\mathcal{N}\left(0,\frac{1-\rho}{\rho^{2}}\right).
 \]

 By the $\delta$-Method with $g(x)=1/x$, $|g'(x)|=1/x^{2}$ we obtain
 \begin{align*}
 \sqrt{N}\left(g\left(\frac{\sum_{i=1}^{N}C_{i,N}-1}{N-1}\right)-g\left(\frac{1}{\rho_{N}}\right)\right) & =\sqrt{N}\left(\frac{N-1}{\sum C_{i,N}-1}-\rho_{N}\right)\\
  & \overset{d}{\rightarrow}\mathcal{N}\left(0,\frac{1-\rho}{\rho^{2}}\cdot\rho^{4}\right)=\mathcal{N}\left(0,\left(1-\rho\right)\rho^{2}\right).\qedhere
 \end{align*}
 \end{proof}

\begin{thm}
The estimator 
\begin{equation*}
\hat{\rho}_{N,T}=\prod_{t=1}^{T}\frac{1}{N}\sum_{k=1}^{N}c_{t}^{k}\cdot\frac{N-1}{\sum_{k=1}^{N}C_{k,N}-1}.
\end{equation*}
is unbiased for $p(y_{1:T})$, i.e.
\begin{equation*}
	\mathbb{E}\left[\hat{\rho}_{N,T}\right] = p(y_{1:T}).
\end{equation*}
\end{thm}

\begin{proof}
The main argument is that the standard proof for unbiasedness using backward induction can be adapted to include our estimator using the geometric random variables $C_{1,N}, \ldots, C_{N,N}$. 
Denote by $X, \overline{X}$ the random variables before and after resampling respectively. 
Then we have
\begin{align*}
& \mathbb{E}\left[ \frac{1}{N} \sum_{i=1}^N w_T^k \mid {X}_{T-1}^{1:N}\right] \\
& = \mathbb{E}\left[\frac{1}{N} \sum_{i=1}^N g(y_{T}\mid {X}_{T}^i ) \mid {X}_{T-1}^{1:N}\right] \\
& = \frac{1}{N} \sum_{i=1}^N \int g(y_{T}\mid {x}_{T}^i ) p(x_T^i \mid {X}_{T-1}^{1:N}) dx_T^i\\
& = \frac{1}{N} \sum_{i=1}^N\int  g(y_T \mid x_{T}^i) \sum_{j=1}^N f(x_{T}^i\mid {X}_{T-1}^j) \frac{g(y_{T-1} \mid {X}_{T-1}^j)}{\sum_{k=1}^Ng(y_{T-1} \mid {X}_{T-1}^k)} d{x}_{T}^i\\
& = \frac{1}{N} \sum_{i=1}^N\sum_{j=1}^N \int  g(y_T \mid x_{T}^i) f(x_{T}^i\mid {X}_{T-1}^j) \frac{g(y_{T-1} \mid {X}_{T-1}^j)}{\sum_{k=1}^Ng(y_{T-1} \mid {X}_{T-1}^k)} d{x}_{T}^i\\
& = \frac{1}{N} \sum_{i=1}^N\sum_{j=1}^N p(y_T \mid X_{T-1}^j) \frac{g(y_{T-1} \mid {X}_{T-1}^j)}{\sum_{k=1}^Ng(y_{T-1} \mid {X}_{T-1}^k)} \\
& = \frac{\sum_{j=1}^N p(y_{T-1:T} \mid {X}_{T-1}^j)}{\sum_{k=1}^N g(y_{T-1} \mid {X}_{T-1}^k)}.
\end{align*}
Now, considering the case where we estimate the weights using the geometric random variables, we have
\begin{align*}
& \mathbb{E}\left[\frac{1}{N}\sum_{k=1}^{N}c_{T}^{k}\cdot\frac{N-1}{\sum_{k=1}^{N}C_{k,N, T}-1}\mid X_{T-1}^{1:N}\right] \\
& = \mathbb{E}\left[\frac{1}{N}\sum_{k=1}^{N}c_{T}^{k}\mathbb{E}\left\{\frac{N-1}{\sum_{k=1}^{N}C_{k,N, T}-1}\mid {X}_T^{1:N}, \overline{X}_{T-1}^{1:N}\right\}\mid X_{T-1}^{1:N}\right] \\
& = \mathbb{E}\left[\frac{1}{N}\sum_{k=1}^{N}c_{T}^{k}
\frac{\sum_{k=1}^{N}w_{T}^{k}}{\sum_{k=1}^{N}c_{T}^{k}}
\mid X_{T-1}^{1:N}\right] \\
& = \mathbb{E}\left[\frac{1}{N}\sum_{k=1}^N w_T^k\mid X_{T-1}^{1:N}\right] \\
& = \frac{\sum_{k=1}^N p(y_{T-1:T} \mid x_{T-1}^k)}{\sum_{k=1}^N g(y_{T-1} \mid x_{T-1}^k)}
\end{align*}
and 
\begin{align*}
& \mathbb{E}\left[\left(\frac{1}{N}\sum_{j=1}^{N}c_{T-1}^{j}\cdot\frac{N-1}{\sum_{k=1}^{N}C_{k,N, T-1}-1}\right) \cdot \left(\frac{1}{N}\sum_{j=1}^{N}c_{T}^{j}\cdot\frac{N-1}{\sum_{k=1}^{N}C_{k,N, T}-1}\right) \mid X_{T-2}^{1:N} \right] \\
& = \mathbb{E}\left[\mathbb{E}\left[\left(\frac{1}{N}\sum_{j=1}^{N}c_{T-1}^{j}\cdot\frac{N-1}{\sum_{k=1}^{N}C_{k,N, T-1}-1}\right) \cdot \left(\frac{1}{N}\sum_{j=1}^{N}c_{T}^{j}\cdot\frac{N-1}{\sum_{k=1}^{N}C_{k,N, T}-1}\right) \mid X_{T-2:T-1}^{1:N}, \overline{X}_{T-2}^{1:N}, C_{1:N, N, T-1}\right] \mid X_{T-2}^{1:N} \right] \\
& = \mathbb{E}\left[\left(\frac{1}{N}\sum_{j=1}^{N}c_{T-1}^{j}\cdot\frac{N-1}{\sum_{k=1}^{N}C_{k,N, T-1}-1}\right)\mathbb{E}\left[\left(\frac{1}{N}\sum_{j=1}^{N}c_{T}^{j}\cdot\frac{N-1}{\sum_{k=1}^{N}C_{k,N, T}-1}\right) \mid X_{T-2:T-1}^{1:N}, \overline{X}_{T-2}^{1:N}, C_{1:N, N, T-1}\right] \mid X_{T-2}^{1:N} \right] \\
& = \mathbb{E}\left[\left(\frac{1}{N}\sum_{j=1}^{N}c_{T-1}^{j}\cdot\frac{N-1}{\sum_{k=1}^{N}C_{k,N, T-1}-1}\right)\frac{\sum_{k=1}^N p(y_{T-1:T} \mid X_{T-1}^k)}{\sum_{k=1}^N g(y_{T-1} \mid X_{T-1}^k)} \mid X_{T-2}^{1:N} \right] \\
& = \mathbb{E}\left[\mathbb{E}\left[\left(\frac{1}{N}\sum_{j=1}^{N}c_{T-1}^{j}\cdot\frac{N-1}{\sum_{k=1}^{N}C_{k,N, T-1}-1}\right)\mid X_{T-1}^k, \overline{X}_{T-2}^k\right]\frac{\sum_{k=1}^N p(y_{T-1:T} \mid X_{T-1}^k)}{\sum_{k=1}^N g(y_{T-1} \mid X_{T-1}^k)}  \mid X_{T-2}^{1:N} \right] \\
& = \mathbb{E}\left[\left(\frac{1}{N}\sum_{k=1}^{N}w_{T-1}^{k}\right) \cdot\frac{\sum_{k=1}^N p(y_{T-1:T} \mid X_{T-1}^k)}{\sum_{k=1}^N g(y_{T-1} \mid X_{T-1}^k)} \mid X_{T-2}^{1:N}\right].
\end{align*}

For the final expectation, we can calculate
\begin{align*}
&\mathbb{E}\left[\frac{1}{N} \sum_{k=1}^N w_{T-1}^k \cdot \frac{\sum_{k=1}^N p(y_{T-1:T} \mid X_{T-1}^k)}{\sum_{k=1}^N g(y_{T-1} \mid X_{T-1}^k)}\mid X_{T-2}^{1:N}\right] \\
& = \mathbb{E}\left[\frac{1}{N}\sum_{k=1}^N g(y_{T-1} \mid X_{T-1}^k)\frac{\sum_{k=1}^N p(y_{T-1:T} \mid X_{T-1}^k)}{\sum_{k=1}^N g(y_{T-1} \mid X_{T-1}^k)}\mid X_{T-2}^{1:N}\right] \\
& = \mathbb{E}\left[\frac{1}{N}\sum_{k=1}^N p(y_{T-1:T} \mid X_{T-1}^k)\mid X_{T-2}^{1:N}\right] \\
& = \frac{1}{N} \sum_{k=1}^N\int p(y_{T-1:T} \mid x_{T-1}^k) p(x_{T-1}^k \mid X_{T-2}^{1:N}) d x_{T-1}^k \\
& = \frac{1}{N} \sum_{k=1}^N\int p(y_{T-1:T} \mid x_{T-1}^k) \frac{\sum_{j=1}^N g(y_{T-2} \mid X_{T-2}^j) f(x_{T-1}^k \mid X_{T-2}^j)}{\sum_{j=1}^N g(y_{T-2} \mid X_{T-2}^j)}d x_{T-1}^k \\
& = \sum_{j=1}^N p(y_{T-1:T} \mid X_{T-2}^j) \frac{g(y_{T-2} \mid X_{T-2}^j) }{\sum_{k=1}^N g(y_{T-2} \mid X_{T-2}^k)} = \frac{\sum_{j=1}^N p(y_{T-2:T} \mid X_{T-2}^j)}{\sum_{j=1}^Ng(y_{T-2} \mid X_{T-2}^j)}. 
\end{align*}
Repeated application of these steps yields
\begin{align*}
\mathbb{E}\left[\prod_{t=1}^{T}\frac{1}{N}\sum_{k=1}^{N}c_{t}^{k}\cdot\frac{N-1}{\sum_{k=1}^{N}C_{k,N}-1}\right] & = \mathbb{E}\left[\frac{1}{N} \sum_{k=1}^N g(y_1 \mid X_1^k) \frac{\sum_{j=1}^N p(y_{1:T} \mid X_{1}^j)}{\sum_{j=1}^N g(y_1 \mid X_{1}^j)}\right]\\
& = p(y_{1:T}).\qedhere
\end{align*}
\end{proof}
 
\section{Further Details to the Applications}

\subsection{Locally optimal proposal: run-time}

We conjecture that the advantage of the BRPF over the RWPF grows as the state transition get computationally more expensive. We will demonstrate that on a toy example.

First note that in case of the Gaussian state space model, the locally optimal proposal,
\begin{align*}
q^*(x_{t}\mid y_{t},x_{t-1}) & \propto g(y_{t}\mid x_{t})f(x_{t}\mid x_{t-1})\\
 & \propto\varphi\left(x_{t};\frac{1}{2}(ax_{t-1}+y_{t}),2.5^{2}\right)
\end{align*}
is known analytically and is Gaussian. Therefore, in this special case the rejection sampler can be avoided. 
This will not usually be the case, hence we use the rejection sampler for our simulation. However, this will 
significantly speed up the state transition. In \autoref{fig:test} we compare the run-time of both, the RWPF and BRPF for 
different numbers of particles. In scenario (a) we use the cheap transition. Here the RWPF is very efficient. However,
in scenario (b), which we also show in the main paper, where sampling from the transition is more expensive, we can implement the BRPF with the same run-time as the RWPF.

\begin{figure}
\centering
\begin{subfigure}{.5\textwidth}
  \centering
  \includegraphics[width=.9\linewidth]{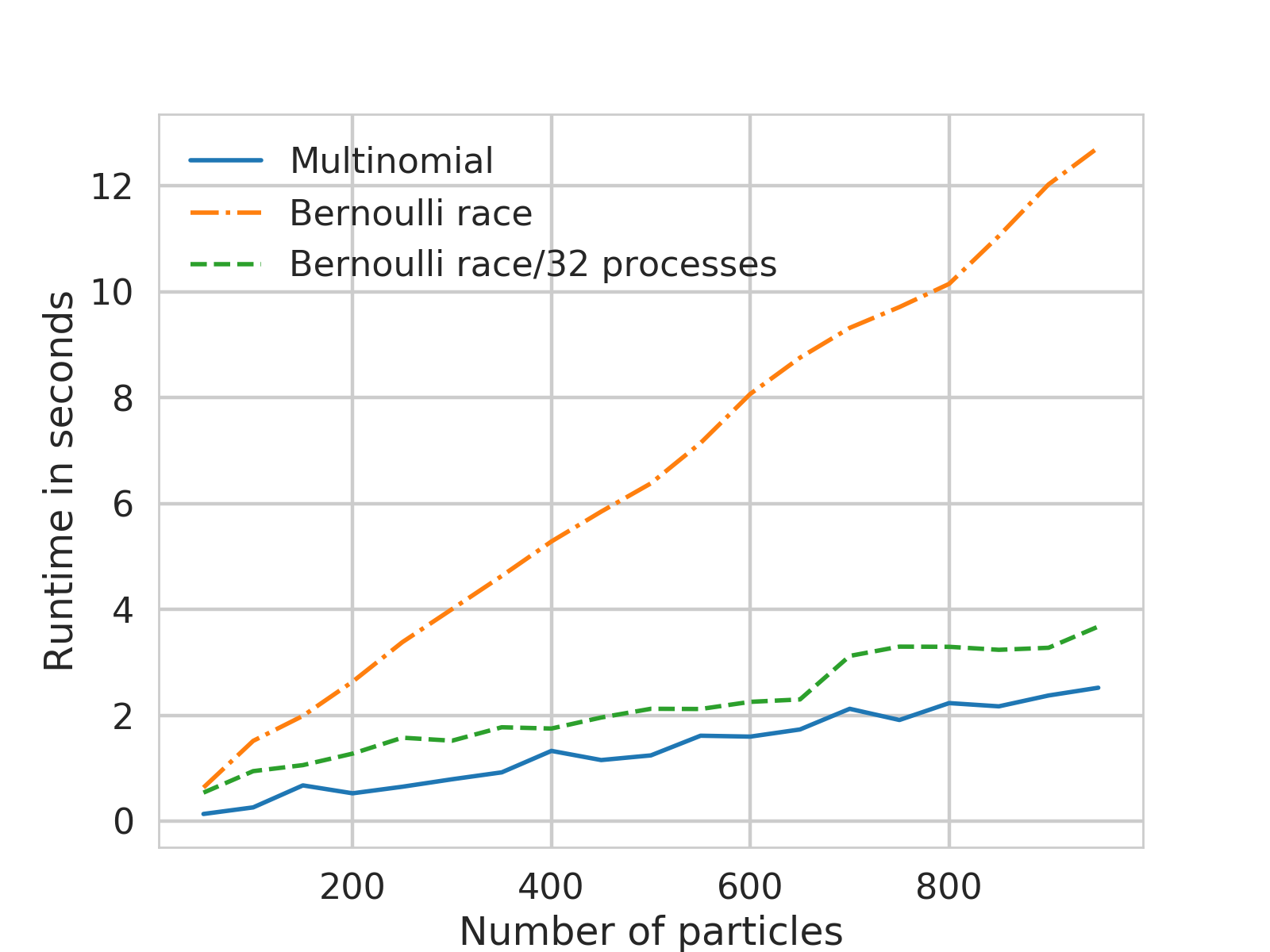}
  \caption{Simulating from known Gaussian.}
  \label{fig:sub1}
\end{subfigure}%
\begin{subfigure}{.5\textwidth}
  \centering
  \includegraphics[width=.9\linewidth]{kf_complexity_rej.png}
  \caption{Rejection Sampler.}
  \label{fig:sub2}
\end{subfigure}
\caption{Comparing the run-time for different implementation for the proposal $q(\cdot \mid x, y)$. In case (a) the Gaussian proposal is sampled using a na\"ive implementation. In (b) we use a rejection sampler proposing from the state transition.}
\label{fig:test}
\end{figure}

\subsection{Sine diffusion: run-time}

Here we compare the run-time for the RWPF and BRPF for the sine diffusion state space model. As pointed out in the main
text, the BRPF is slower when implemented sequentially, but we observe in \autoref{fig:sine_comp} that the difference almost
completely vanishes when use a parallel implementation with 16 processes.

\begin{figure}
\centering
\includegraphics[width=0.5\linewidth]{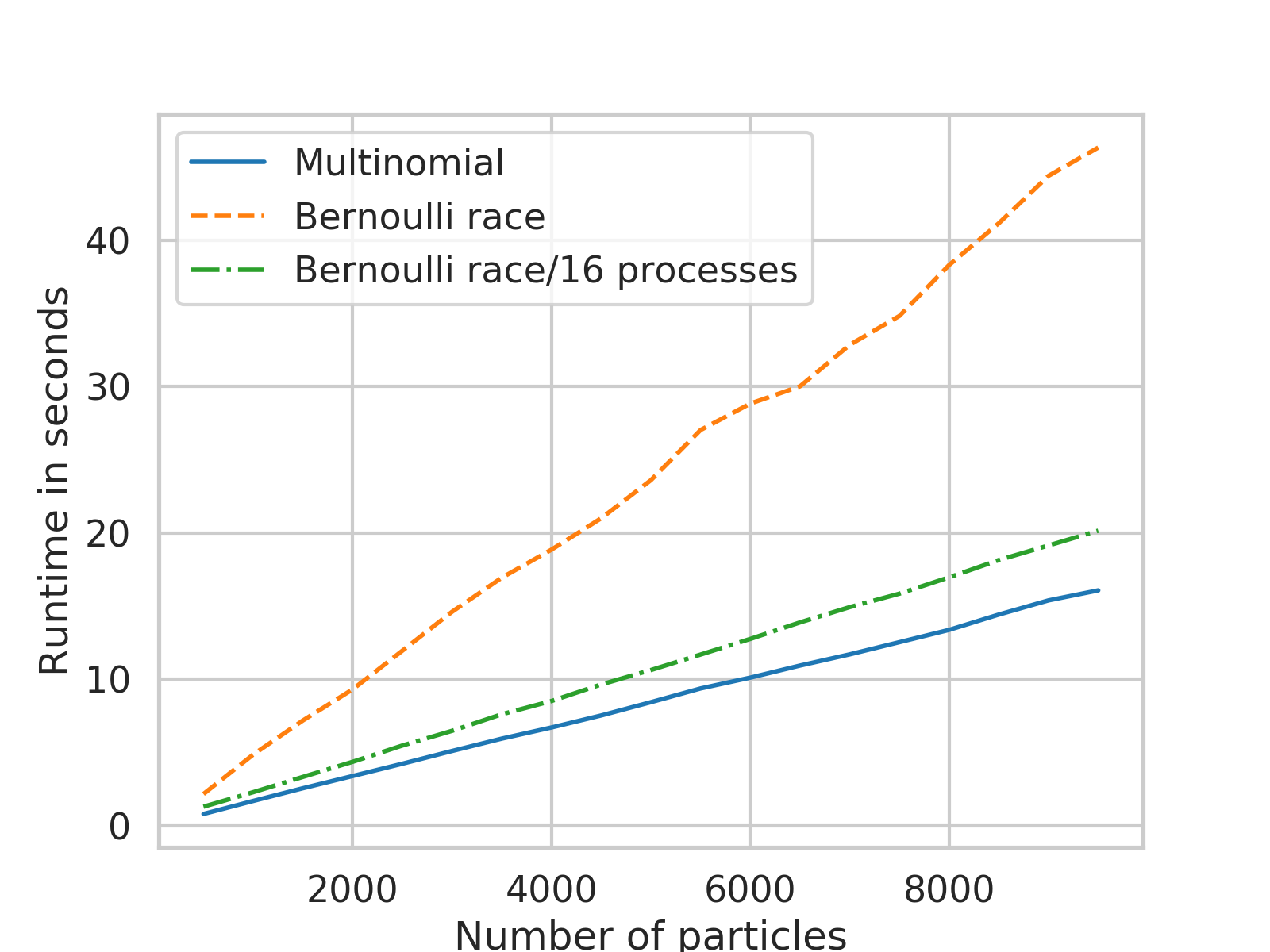}
\caption{Comparison of the run-time for RWPF and BRPF for sine diffusion state space model. For BRPF we show a sequential 
and a parallel implementation with 16 cores. For all algorithms we show wall-clock time for number of particles $N$. }
\label{fig:sine_comp}
\end{figure}

\section{Cox Process inference}

To further demonstrate the range of possible application and for further illustration of the Bernoulli race particle filter, we consider another application of the BRPF to a Cox process whose
intensity function is governed by a Gaussian process that is normalized through a sigmoid function. The underlying Gaussian process is modelled as a Gauss--Markov process, allowing us to perform inference sequentially as done by \textcite{li2018sequential} based on models of \textcite{christensen2012forecasting, adams2009tractable}.

\subsection{Likelihood and Estimation}

We assume the latent Gaussian process to evolve according to the Ornstein--Uhlenbeck process
\begin{equation}\label{eq:ou_process}
dX_{t}=AX_{t}dt+h dB_{t} \qquad t \in[0, \mathcal{T}],
\end{equation}
where $(B_t)_{t\in[0,\mathcal{T}]}$ denotes a Brownian motion from $0$ to $\mathcal{T}$. We will describe the structural assumptions on $A$ and $h$ in more detail below.
Denote the observed data by $s_{1:n}\subset [0, \mathcal{T}]$, where $n$ is the number of observations. We model the data by a Cox process with intensity function 
\begin{equation*}
\lambda(t)=\sigma(X_{1, t})=\frac{\lambda_\mathrm{max}}{1+\exp(-X_{1, t})}=\lambda_\mathrm{max}\frac{\exp(X_{1, t})}{1+\exp(X_{1, t})},
\end{equation*}
where $(X_{1, t})_{t\in[0, \mathcal{T}]}$ is the first coordinate of the SDE \eqref{eq:ou_process} introduced before. The likelihood function of the Cox process conditional on the intensity is Poisson with 
\begin{equation}\label{eq:poisson_likelihood}
p(s_{1:n}\mid\lambda,\mathcal{T})=\exp\left(-\int_0^{\mathcal{T}}\lambda(s)ds\right)\prod_{i=1}^{n}\lambda(s_{i}).
\end{equation}
The likelihood is clearly intractable due to the integral which involves the latent Gaussian process. To implement a particle filter for sequential inference we need to find an unbiased coin with probability proportional to the integral in \eqref{eq:poisson_likelihood}.
For any interval $[t_0,t_1] \subset [0, \mathcal{T}]$ we can find an estimator for $\exp(-\int_{t_0}^{t_1}\lambda(s)ds)$
by observing that for $U\sim \mathrm{Unif}[t_0,t_1]$ and a fixed process
$X = (X_{t})_{t\in[t_0,t_1]}$ we have
\begin{align*}
\mathbb{E}\left[\lambda(U) \mid X = x\right] & = \int_{t_0}^{t_1}\frac{\lambda(s)}{t_1-t_0}ds \\
& = \frac{1}{t_0-t_1}\int_{t_0}^{t_1}\frac{1}{1 + \exp(-x_s)}ds
\end{align*}
and hence, if $V\sim \mathrm{Unif}[0, 1]$ independent of all other random variables, we obtain
\[
P\left(V\leq\lambda(U) \mid X=x\right)=\int_{t_0}^{t_1}\frac{\lambda(s)}{t_1-t_0}ds.
\]
Now sample $K\sim \mathrm{Pois}(\lambda_\mathrm{max}(t_1-t_0))$, then an unbiased coin flip can be generated using
\begin{equation*}
	Z = \prod_{i=1}^{K}1\left\{ V_{i}\leq \frac{\lambda_\mathrm{max} - \lambda(U_i)}{\lambda_\mathrm{max}}\right\},
\end{equation*}
where $1\{A\}$ denotes the indicator function of the set $A$. This estimator is indeed unbiased as can be seen by
\begin{align*}
 & \mathbb{E}\left[Z\right] = \mathbb{E}\left[\prod_{i=1}^{K}1\left\{ V_{i}\leq \frac{\lambda_\mathrm{max} - \lambda(U_i)}{\lambda_\mathrm{max}}\right\} \mid X=x\right] \\
 & =e^{-\lambda_\mathrm{max}(t_1-t_0)}\sum_{k=0}^{\infty}\frac{\left(\lambda_\mathrm{max}(t_1-t_0)\right)^{k}}{k!}\prod_{i=1}^{k}\mathbb{E}\left[1\left\{ V_{i}\leq \frac{\lambda_\mathrm{max}-\lambda(U_i)}{\lambda_\mathrm{max}}\right\} \mid K=k,X=x\right]\\
 & =e^{-\lambda_\mathrm{max}(t_1-t_0)}\sum_{k=0}^{\infty}\frac{\lambda_\mathrm{max}^k\left(t_1-t_0\right)^{k}}{k!}\left(\int_{t_0}^{t_1}\frac{\lambda_\mathrm{max}-\lambda(s)}{\lambda_\mathrm{max}(t_1-t_0)}ds\right)^{k}\\
 & =e^{-\lambda_\mathrm{max}(t_1-t_0)}\sum_{k=0}^{\infty}\frac{\left(\int_{t_0}^{t_1}\lambda_\mathrm{max}-\lambda(s)ds\right)^{k}}{k!\lambda_\mathrm{max}^k}\\
 & =e^{-\lambda_\mathrm{max}(t_1-t_0)}e^{\lambda_\mathrm{max}(t_1-t_0)}\exp\left(-\int_{t_0}^{t_1}\lambda(s)ds\right)\\
 & =\exp\left(-\int_{t_0}^{t_1}\lambda(s)ds\right).
\end{align*}
To implement the Bernoulli race particle filter we sample from the prior, which can be done exactly, as seen in the next section.
As we are sampling from the prior the particle weights just depend on the observation likelihood. Suppose the particle filter is at current time $t_0$ and we propose to move to $t_1 = t_0 + \Delta t$. Let $(X_{1, t_1}^k), k=1,\ldots, N$ denote a proposed particle path. Then the weight is given by 
\begin{equation*}
	g(s_{i\colon s_i \in [t_0, t_1]} \mid X^k) = \exp\left(-\int_{t_0}^{t_1} \lambda^k(s) ds \right) \prod_{i\colon s_i \in [t_0, t_1]} \lambda(s^k_i),
\end{equation*}
where 
\begin{equation*}
	\lambda^k(s) = \frac{\exp(X_{1, s}^k)}{1 + \exp(X_{1, s}^k)}.	
\end{equation*}
In the notation of Section 4 we have
\begin{align*}
	c_{t_1}^k &= \prod_{i\colon s_i \in [t_0, t_1]} \lambda^k(s_i), \\
	b_{t_1}^k &= \exp\left(-\int_{t_0}^{t_1} \lambda^k(s) ds\right).
\end{align*}
To implement the estimator $\hat{Z}$, we need to be able to simulate a bridge from the stochastic differential equation \eqref{eq:ou_process}. Since the process at hand is analytically tractable, exact sampling from the bridge is straightforward, see e.g. \textcite{bladt2016simulation} using the quantities computed below.  

\subsection{Ornstein--Uhlenbeck Prior}

In this section we provide further details on the assumptions on the Gaussian process and show how we can sample the state transition.
We assume the prior distribution on the latent space is given by the Ornstein--Uhlenbeck process \eqref{eq:ou_process}, that is
\[
dX_{t}=AX_{t}dt+ h dB_{t} \qquad t \in [0, \mathcal{T}]. 
\]
For $A$ and $h$ we take the values
\begin{align*}
A & =\left[\begin{array}{cc}
0 & 1\\
0 & \theta
\end{array}\right], \quad 
h =\left[\begin{array}{c}
0\\
\sigma
\end{array}\right]
\end{align*}
where $\theta$ is a negative real value.
The process can be written as a differential equation with random velocity following a mean reverting real-valued Ornstein--Uhlenbeck process,
\begin{align*}
	\begin{pmatrix}
		dX_{1, t} \\
		dX_{2, t} \\
	\end{pmatrix}
	=
	\begin{bmatrix}
		0 & 1 \\
		0 & \theta \\
	\end{bmatrix}
	\begin{pmatrix}
	 X_{1, t} \\
	 X_{2, t} \\
	\end{pmatrix}
	dt +
	\begin{pmatrix}
	0 \\ \sigma
	\end{pmatrix}
	dW_t.
\end{align*}
This can be written as
\begin{align*}
	dX_{1, t} &= X_{2, t}dt \\
	dX_{2, t} &= \theta X_{2, t}dt + \sigma dB_t.
\end{align*}
This is a linear Gaussian system and thus the solution can be derived analytically, see e.g. \textcite{christensen2012forecasting}. The discretized system at time $s > r$ can be simulated exactly by sampling
\begin{equation*}
	X_s \mid (X_r=x) \sim \mathcal{N}\left(e^{A(s-r)x}, e^{A(s-r)}Q(r,s) \left(e^{A(s-r)}\right)^T\right),
\end{equation*}
where 
\begin{equation*}
Q(r,s)=\int_{r}^{s}\exp(-At)hh^{T}\exp(-At)^{T}dt.
\end{equation*}
In order to implement the Ornstein--Uhlenbeck process we need the quantities
\begin{align*}
\exp(At) & =\sum_{k=0}^{\infty}t^{k}\frac{A^{k}}{k!}.
\end{align*}
Note that
\[
A^{2}=\left[\begin{array}{cc}
0 & 1\\
0 & \theta
\end{array}\right]\left[\begin{array}{cc}
0 & 1\\
0 & \theta
\end{array}\right]=\left[\begin{array}{cc}
0 & \theta\\
\text{0} & \theta^{2}
\end{array}\right]
\]
\begin{align*}
A^{k}=A\cdot A^{k-1} & =\left[\begin{array}{cc}
0 & 1\\
\text{0} & \theta
\end{array}\right]\left[\begin{array}{cc}
0 & \theta^{k-2}\\
\text{0} & \theta^{k-1}
\end{array}\right]\\
 & =\left[\begin{array}{cc}
0 & \theta^{k-1}\\
\text{0} & \theta^{k}
\end{array}\right]
\end{align*}
Hence, we can compute the matrix exponential analytically 
\begin{align*}
	\exp(At) &= \sum_{k=0}^\infty \frac{t^k}{k!} A^k \\
					 &= \sum_{k=0}^\infty \frac{t^k}{k!}
	\begin{bmatrix}
		0 & \theta^{k-1} \\
		0 & \theta^k \\
	\end{bmatrix}
	\\
	& = 
	\begin{bmatrix}
	1 & 0 \\
	0 & 1 
	\end{bmatrix}
	+ \sum_{k=1}^\infty \frac{t^k\theta^k}{k!}
	\begin{bmatrix}
		0 & 1/\theta \\
		0 & 1 \\
	\end{bmatrix} \\
	& = 
	\begin{bmatrix}
	1 & 0 \\
	0 & 1 
	\end{bmatrix}
	 + 
	\begin{bmatrix}
	0 & \exp(\theta t)/\theta - 1/\theta \\
	0 & \exp(\theta t) - 1\\
	\end{bmatrix}
	\\
	& = 
	\begin{bmatrix}
	1 & \exp(\theta t)/\theta - 1/\theta \\
	0 & \exp(\theta t)\\
	\end{bmatrix}.
\end{align*}
This gives us 
\begin{align*}
	\exp(-At)h & = 
	\begin{bmatrix}
	1 & \exp(-\theta t)/\theta - 1/\theta \\
	0 & \exp(-\theta t)
	\end{bmatrix} 
	\begin{bmatrix}
	0 \\
	\sigma 
	\end{bmatrix} \\
	& = 
	\begin{bmatrix}
	\frac{\sigma}{\theta} \exp(-\theta t) - \sigma/\theta \\
	\sigma \exp(-\theta t)
	\end{bmatrix}
\end{align*}
and
\begin{align*}
& \exp(-At)hh^T\exp(-At) \\
& =
	\begin{bmatrix}
	\frac{\sigma}{\theta} \exp(-\theta t) - \sigma/\theta \\
	\sigma \exp(-\theta t)
	\end{bmatrix}
	\begin{bmatrix}
	\frac{\sigma}{\theta} \exp(-\theta t) - \sigma/\theta & 
	\sigma \exp(-\theta t)
	\end{bmatrix} \\
& = 
	\begin{bmatrix}
		\frac{\sigma^2}{\theta^2}\left(\exp(-\theta t)-1 \right)^2 & \frac{\sigma^2}{\theta}\left(\exp(-2\theta t)-\exp(-\theta t)\right) \\
		\frac{\sigma^2}{\theta}\left(\exp(-2\theta t)-\exp(-\theta t) \right) & \sigma^2 \exp(-2\theta t)
	\end{bmatrix}.
\end{align*}
Thus 
\begin{align*}
Q(r,s) & =\int_{r}^{s}\exp(-At)hh^{T}\exp(-At)^{T}dt \\
			 & =
	\begin{bmatrix}
		\int_r^s\sigma^2\left(\exp(-\theta t)-1 \right)^2 dt & \int_r^s \frac{\sigma^2}{\theta}\left(\exp(-2\theta t)-\exp(-\theta t) \right)dt \\
		\int_r^s\sigma^2\left(\exp(-2\theta t)-\exp(-\theta t)/\theta \right)dt & \int_r^s \sigma^2 \exp(-2\theta t) dt
	\end{bmatrix} \\
	     & = 
	\begin{bmatrix}
		\frac{\sigma^2}{\theta^3}\left(-2\theta r + e^{-2\theta r}-4e^{-\theta r}-e^{-2\theta s}+4e^{-\theta s}+2\theta s \right)	& \frac{\sigma^2}{2\theta^2} \left(e^{-2r\theta} - e^{-2s\theta}\right) - \frac{\sigma^2}{\theta^2} \left(e^{-r\theta} - e^{-s\theta}\right) \\
\frac{\sigma^2}{2\theta^2} \left(e^{-2r\theta} - e^{-2s\theta}\right) - \frac{\sigma^2}{\theta^2} \left(e^{-r\theta} - e^{-s\theta}\right)  & \frac{\sigma^2}{2\theta} \left(e^{-2r\theta}- e^{-2s\theta}\right)
	\end{bmatrix}
\end{align*}
The covariance matrix for the system transition 
\begin{align*}
	\mathrm{Cov}(r,s) = e^{A(s-r)} Q(r, s) \left(e^{A(s-r)}\right)^T
\end{align*}
can thus be computed analytically.

\subsection{Simulation}

We simulate a non-homogeneous Poisson process in the interval $[0, 50]$ using the intensity function 
\begin{equation*}
	\lambda(s) = 2\exp(-s/15) + \exp\left(-((s - 25)/10)^2 \right),
\end{equation*}
see e.g. \textcite{adams2009tractable, li2018sequential}. We run the Bernoulli race particle filter with the above specification and $30$ particles to find the intensity function on the interval $[0, 50]$ which is divided into 10 equispaced intervals. The result of one run is shown in \autoref{fig:appendix_1}. We can see that the Bernoulli race particle filter is approximating the true intensity function.

\begin{figure}
\centering
\includegraphics[width=\textwidth]{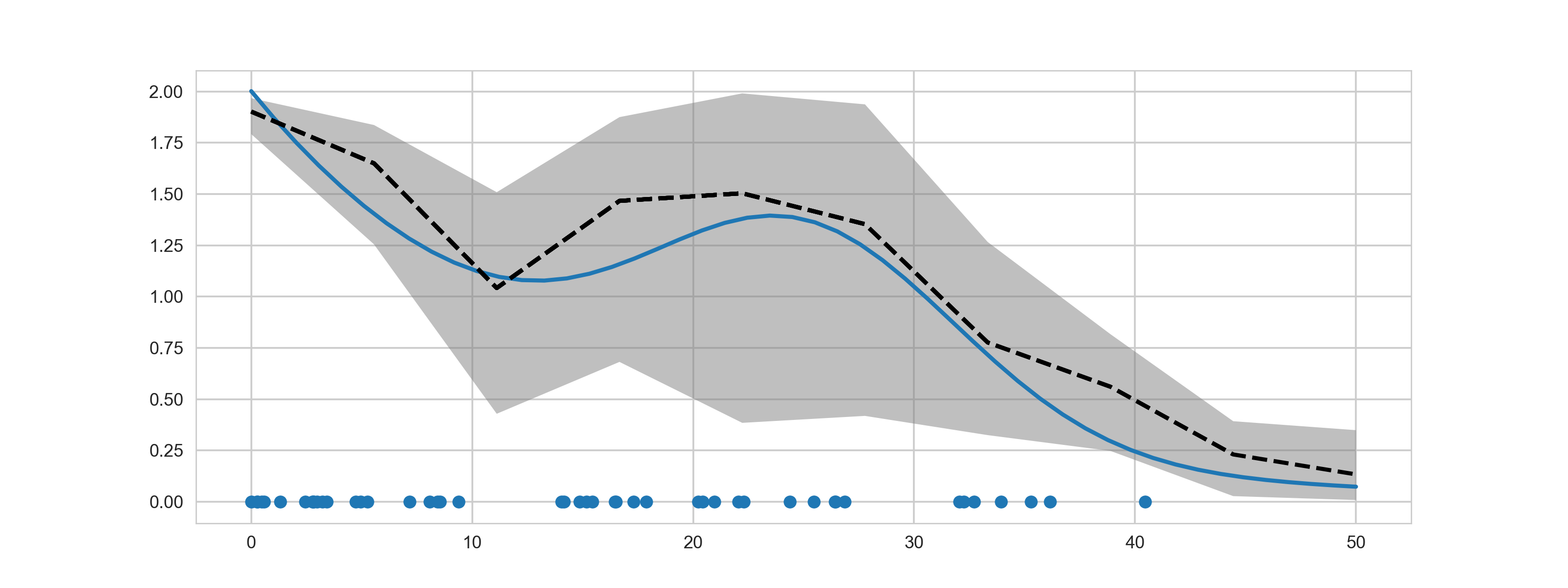}
\caption{\label{fig:appendix_1} Result of one run of the Bernoulli race particle filter using 30 particles. Arrival times of the Poisson process are shown as dots on the abscissa. The solid blue line shows the true intensity function and the dashed line shows the mean of the particle approximation at times $t_1, t_2, \ldots$. The shaded area highlights the 10\% and 90\% quantile of the particle approximation.}
\end{figure}

\printbibliography[heading=subbibliography] 
\end{refsection}

\end{document}